\documentclass{amsart}
\usepackage{graphicx}
\usepackage{amsmath}
\usepackage{enumerate}
\newtheorem{thm}{Theorem}
\newtheorem{cor}[thm]{Corollary}

\theoremstyle{definition}
\newtheorem{cond}[thm]{Condition}
\theoremstyle{remark}
\newtheorem{rem}{Remark}
\numberwithin{equation}{section}
\numberwithin{thm}{section}

\newcommand{\set}[1]{\left\{#1\right\}}

\newcommand{\RR}{\mathbb{R}}

\newcommand{\Rplus}{\mathbb{R}_{\ge 0}}

\newcommand{\binv}{b_\textit{inv}}
\newcommand{\basymp}{b_\textit{asymp}}
\newcommand{\bnorm}{b_\textit{norm}}
\newcommand{\bynorm}{b_\textit{y-norm}}
\newcommand{\bfwnorm}{b_\textit{fw-norm}}

\begin{document}

\title[Erratum to: Yield curve shapes in affine models]{Erratum to: `Yield curve shapes and the asymptotic short rate distribution in affine one-factor models'}%
\author{Martin Keller-Ressel}%


\maketitle
\begin{abstract}
I would like to thank Ralf Korn for alerting me to an error in the original paper \cite{KR08}. The error concerns the threshold at which the yield curve in an affine short-rate model changes from normal (strictly increasing) to humped (endowed with a single maximum). In particular, it is not true that this threshold is the same for the forward curve and for the yield curve, as claimed in \cite{KR08}. Below, the correct mathematical expression for the threshold is given, supplemented with a self-contained and corrected proof.
\end{abstract}

\section{Setting}
In \cite{KR08} affine short rate models for bond pricing were considered, i.e. models where the risk-neutral short rate process $r = (r_t)_{t \geq 0}$ is given by an affine process in the sense of \cite{DFS03}. The process $r$ takes values in a state space $D$, which is either $\Rplus$ or $\RR$. In this setting, the price at time $t$ of a zero-coupon bond with time-to-maturity $x$, denoted by $P(t,t+x)$ is of the form
\[P(t,t+x) = \exp\left(A(x) + r_t B(x)\right),\]
where $A$ and $B$ satisfy the generalized Riccati differential equations
\begin{equation}\label{eq:riccati}
\begin{aligned}
\partial_x A(x) &= F(B(x)), &\qquad A(0) &= 0\\
\partial_x B(x) &= R(B(x))  -1, &\qquad B(0) &= 0.
\end{aligned}
\end{equation}
The functions $F$ and $R$ are of L\'evy-Khintchine-form and their parameterization is in one-to-one correspondence with the infinitesimal generator of $r$, cf. \cite[Sec.~2]{KR08}. Derived from the bond price, are the \textbf{yield curve}
\[Y(x,r_t) := -\frac{\log P(t,t+x)}{x} = -\frac{A(x)}{x} - r_t \frac{B(x)}{x}\]
and the \textbf{forward curve}
\[f(x,r_t) := -\partial_x \log P(t,t+x) = -A'(x) - r_t B'(x).\]
The first objective of \cite{KR08} was to derive the long-term yield and long-term forward rate. It was shown that the equation $R(c) = 1$ has at most a single negative solution $c$, and that under mild conditions
\[\basymp := \lim_{x \to \infty} Y(x,r_t) = \lim_{x \to \infty} f(x,r_t) = -F(c)\]
if such a solution exists, cf. \cite[Thm.~3.7]{KR08}. We remark that $\lambda := -\frac{1}{c} > 0$ was called \textbf{quasi-mean-reversion} of $r$ in \cite{KR08}, with the convention that $\lambda = 0$ if no negative solution $c$ exists. 
The second objective of \cite{KR08} was to characterize all possible shapes of the yield and the forward curve. Recall that in common terminology, the yield or the forward curve is called 
\begin{itemize}
\item \textbf{normal} if it is a strictly increasing function of $x$,
\item \textbf{inverse} if it is a strictly decreasing function of $x$,
\item \textbf{humped} if it has exactly one local maximum and no local minimum in $(0,\infty)$.
\end{itemize}
Finally, we recall the technical condition \cite[Condition~3.1]{KR08}, in slightly rephrased form. The condition is necessary to guarantee finite bond prices, when negative values of the short-rate are allowed.
\begin{cond}\label{cond1}
We assume that $r$ is regular and conservative. If $r$ has state space $D = \RR$, which necessarily implies that $R$ is of the linear form $R(x) = \beta x$ (cf. \cite{DFS03}), we require that 
\[F(x) < \infty \quad \text{for all} \quad x \in \begin{cases}(1/\beta,0] &\quad \text{if $\beta < 0$}\\(-\infty,0], &\quad \text{else}.\end{cases}\] 
\end{cond}

\section{Corrections to results}
Theorem~3.1 in \cite{KR08} should be replaced by the following corrected version:
\begin{thm}\label{thm:yield}
Let the risk-neutral short rate process be given by a one-dimensional affine process $(r_t)_{t \geq 0}$ satisfying Condition~\ref{cond1} and with quasi-mean-reversion $-1/c = \lambda > 0$. In addition suppose that $F \neq 0$ and that at least one of $F$ and $R$ is non-linear. Then the following holds:
\begin{enumerate}
\item The yield curve $Y(.,r_t)$ can only be normal, inverse or humped.
\item Define 
\begin{align*}
\bynorm &:= \frac{1}{c}\int_c^{0} \frac{F(u) - F(c)}{R(u) - 1} du, \qquad \\
\binv &:= \begin{cases}-\frac{F'(0)}{R'(0)} &\quad \text{if $R'(0) < 0$}\\+\infty &\quad \text{if $R'(0) \ge 0$}.\end{cases}
\end{align*}
The yield curve is normal if $r_t \le \bynorm$, humped if $\bynorm < r_t < \binv$ and inverse if $r_t \ge \binv$.
\end{enumerate}
\end{thm}
\begin{rem}
The correction only concerns the expression for $\bynorm$, which was called $\bnorm$ in \cite{KR08} and erroneously given as $\bnorm = -F'(c) / R'(c)$. All other parts of the theorem are the same as in \cite[Thm.~3.1]{KR08}.
\end{rem}
Corollary~3.11 in \cite{KR08} should be replaced by the following result:
\begin{thm}\label{thm:fw}
Define $\binv$ as in Thm.~\ref{thm:yield} and set
\[\bfwnorm := -\frac{F'(c)}{R'(c)}.\]
Under the conditions of Theorem~\ref{thm:yield} the following holds:
\begin{enumerate}
\item The forward curve $f(.,r_t)$ can only be normal, inverse or humped.
\item The forward curve is normal if $r_t \le b_\text{fw-norm}$, humped if $b_\text{fw-norm} < r_t < b_\text{inv}$ and inverse if $r_t \ge b_\text{inv}$.
\end{enumerate}
\end{thm}
\begin{rem}
We have intentionally renamed the result from Corollary to Theorem, since the correction changes the logical structure of the proof. Note that the above result is equivalent to \cite[Cor.~3.11]{KR08} up to the notational change from $\bnorm$ to $\bfwnorm$. Note that $\bynorm \neq \bfwnorm$ in general, while in \cite{KR08} it was erroneously claimed that $\bynorm = \bfwnorm$. 
\end{rem}
Corollary~3.12 in \cite{KR08} should be replaced by the following result:
\begin{cor}\label{cor:order}Under the conditions of Theorem~\ref{thm:yield} it holds that
\begin{equation}\label{eq:serial_ineq}
\bfwnorm < \bynorm < \basymp < \binv.
\end{equation}
In addition, the state space $D$ of the short rate process satisfies
\[D \cap (\bynorm, \binv) \neq \emptyset.\]
\end{cor}

The error also affects \cite[Figure~1]{KR08}, where the expression for $\bnorm$ should be replaced by the correct value of $\bynorm$. It also affects the application section \cite[Section~4]{KR08}, where the values of $\bnorm$ and $\binv$ are calculated in different models. The corrections to \cite[Section~4]{KR08} are as follows:\\

In the \textbf{Vasicek model}, the short rate is given by
\[dr_t = -\lambda (r_t - \theta)\,dt + \sigma \,dW_t, \quad r_0 \in \RR,\]
with $\lambda, \theta, \sigma > 0$. This leads to the parameterization
\begin{align}
F(u) &= \lambda \theta u  + \frac{\sigma^2}{2}u^2,\\
R(u) &= -\lambda u.
\end{align}
By direct calculation we obtain
\begin{align}
\bynorm &= \theta - \frac{3 \sigma^2}{4 \lambda^2},\\
\bfwnorm &= \theta - \frac{\sigma^2}{\lambda^2}.
\end{align}
Note that the value of $\bynorm$ is now consistent with the results of \cite[p186]{V77}.\\

In the \textbf{Cox-Ingersoll-Ross model}, the short rate is given by 
\[r_t = - a(r_t - \theta)\,dt + \sigma \sqrt{r_t} \,dW_t, \quad r_0 \in \Rplus,\]
with $a, \theta, \sigma > 0$. This leads to the parameterization
\begin{align}
F(u) &= a \theta u\\
R(u) &= - \frac{\sigma^2}{2}u^2 - au.
\end{align}
By direct calculation we obtain
\begin{align}
\bynorm &= \frac{2 a \theta}{\gamma - a} \log \left(\frac{2\gamma}{a + \gamma}\right)\\
\bfwnorm &= \frac{a \theta}{\gamma}
\end{align}
where $\gamma := \sqrt{2\sigma^2 + a^2}$.\\

In the \textbf{gamma model}, the short rate is given by an Ornstein-Uhlenbeck-type process, driven by a compound Poisson process with intensity $\lambda k$ and exponentially distributed jump heights of mean $1/\theta$, see \cite[Sec.~4.4]{KR08} for details. In this model, we have
\begin{align}
F(u) &= \frac{\lambda \theta k u}{1 - \theta u}\\
R(u) &= - \lambda u.
\end{align}
and by direct calculation we obtain
\begin{align}
\bynorm &= \frac{k \lambda}{1 + \theta/\lambda} \log\left(1 + \theta/\lambda\right)\\
\bfwnorm &= \frac{k \theta}{(1 + \theta/\lambda)^2}.
\end{align}

Since the resulting expressions are quite involved, we omit the calculations for the extended CIR model \cite[Eq.~(4.7)]{KR08}.

\section{Corrected proofs}
To prepare for the corrected proofs, we collect the following properties from \cite[Sec.~2 and 3.1]{KR08}, which hold for the functions $F$, $R$, $B$ and for the state space $D$ under the assumptions of Theorem~\ref{thm:yield}:
\begin{enumerate}[(P1)]
\item $F$ is either strictly convex or linear; the same holds for $R$. Both functions are continuously differentiable on the interior of their effective domain.
\item The function $B$ is strictly decreasing with limit $\lim_{x \to \infty} B(x) = c$.
\item $F(0) = R(0) = 0$ and $R'(c) < 0$. In addition, $F'(0) >0$ if $D = \Rplus$.
\item Either \begin{enumerate}[(a)]
\item $D = \Rplus$; or
\item $D = \RR$ and $R(u) = u/c$ with $c < 0$.
\end{enumerate}
\end{enumerate}
Note that Theorem~\ref{thm:yield} assumes that at least one of $F$ and $R$ is non-linear. Together with (P1) this implies the following: 
\begin{enumerate}[(P1')]
\item At least one of $F$ and $R$ is strictly convex.
\end{enumerate}

In addition, we introduce the following terminology: Let $f: (0,\infty) \to \RR$ be a continuous function. The \emph{zero set} of $f$ is $Z := \set{x \in (0,\infty): f(x) = 0}$. The \emph{sign sequence} of $Z$ is the sequence of signs $\set{+,-}$ that $f$ takes on the complement of $Z$, ordered by the natural order on $\RR$. For example, the function $x^2 - 1$ on $(0,\infty)$ has the finite sign sequence $(-+)$; the function $\sin(x)$ has the infinite sign sequence $(+-+-\dotsm)$. 
An obvious, but important property is the following: Let $g: (0,\infty) \to (0,\infty)$ be a \emph{positive} continuous function. Then $fg$ has the same zero set and the same sign sequence as $f$.

\begin{proof}[Proof of Theorem~\ref{thm:fw}]
From the Riccati equations \eqref{eq:riccati}, we can write the derivative of the forward curve as
\begin{equation}\label{eq:k}
\partial_x f(x,r_t) = - B'(x) \cdot \underbrace{\left\{F'(B(x)) + r_t R'(B(x))\right\}}_{:=k(x)}.
\end{equation}
Note that by (P2) the factor $-B'(x)$ is strictly positive, and hence $\partial_x f$ has the same sign sequence as $k$. We distinguish cases (a) and (b) as in (P4): 
\begin{enumerate}[(a)]
\item Assume that $r_t \in D = \Rplus$. By (P2) $B(x)$ is strictly decreasing and by (P1') either $F'$ or $R'$ is strictly increasing. Thus, if $r_t > 0$, it follows that $k(x)$ is a strictly decreasing function. If $r_t = 0$, then $k$ is either strictly decreasing (if $F'$ is strictly convex) or $k$ is constant (if $F$ is linear). By (P1) these are the only possibilities. In addition, the case $F = 0$ is ruled out by the assumptions.
\item Assume that $r_t \in D = \RR$. In this case $R(u) = u/c$ and hence $R'(u) = 1/c$ is constant and $F'$ is strictly increasing, by (P1'). We conclude that $k$ is strictly decreasing. 
\end{enumerate}
In any case, $k$ is either strictly decreasing or constant and non-zero. Thus the sign sequence of $k$ can be completely characterized by its initial value $k(0)$ and its asymptotic limit as $x$ tends to infinity. Let us first show that 
\begin{equation}\label{eq:k0}
k(0) \le 0  \quad \Longleftrightarrow \quad r_t \ge \binv  =  \begin{cases}-\frac{F'(0)}{R'(0)} &\quad \text{if $R'(0) < 0$}\\+\infty &\quad \text{if $R'(0) \ge 0$}.\end{cases}
\end{equation}
We have $k(0) = F'(0) + r_t R'(0)$, such that the assertion follows immediately if $R'(0) < 0$. Consider the complementary case $R'(0) \ge 0$. This rules out case (b) in (P4) and hence we may assume that $D = \Rplus$. Since $F'(0) > 0$ by (P3), \eqref{eq:k0} follows. Next we show that 
\begin{equation}\label{eq:kinfty}
 \lim_{x \to \infty }k(x) \ge 0  \quad \Longleftrightarrow \quad  r_t \le \bfwnorm = -\frac{F'(c)}{R'(c)} .
\end{equation}
This follows immediately from $\lim_{x \to \infty}k(x) = F'(c) + r_t R'(c)$ and $R'(c) <  0$, by (P3).
Combining \eqref{eq:k0} with \eqref{eq:kinfty}, and using that $k$ is either strictly decreasing or constant and non-zero we obtain
\begin{equation}\label{eq:rprop}
\begin{aligned}
r_t \ge \binv \quad & \Longleftrightarrow \quad  k \text{ has sign sequence } (-) \\
r_t \le \bfwnorm \quad & \Longleftrightarrow \quad k \text{ has sign sequence } (+) \\
r_t \in (\bfwnorm, \binv) \quad & \Longleftrightarrow \quad k \text{ has sign sequence } (+-). 
\end{aligned}
\end{equation}
Since $\partial_x f$ has the same sign sequence as $k$, these statements can be directly translated into monotonicity properties of $f$: In the first case the forward curve $f$ is strictly decreasing, i.e. inverse; in the second case it is strictly increasing, i.e. normal. In the third case it is strictly increasing up to the unique zero of $k$ and then strictly decreasing, i.e. humped. No other cases are possible. 
\end{proof}

\begin{proof}[Proof of Theorem~\ref{thm:yield}]
From the Riccati equations \eqref{eq:riccati}, we can write the derivative of the yield curve as
\[\partial_x Y(x,r_t) = \frac{1}{x^2} (A(x) + r_t B(x)) - \frac{1}{x} \left\{F(B(x)) + r_t \left[R(B(x)) - 1\right]\right\}.\]
Multiplying by the positive function $x^2$ we see that  $\partial_x Y(x,r_t)$ has the same zero set and the same sign sequence as
\[M(x) := \left[A(x) - x F(B(x))\right] + r_t \left\{B(x) - x \left[R(B(x)) - 1\right]\right\}.\]
The derivative of $M$ is given by
\[M'(x) := -x B'(x) \cdot \left\{F'(B(x)) + r_t R'(B(x))\right\} = - xB'(x) \cdot k(x),\]
with $k$ as in \eqref{eq:k}. Note that by (P2) the factor $-x B'(x)$ is strictly positive, and hence $M'$ has the same sign sequence as $k$, which was already analyzed in \eqref{eq:rprop}. Since $M(0) = 0$, we can conclude that 
\begin{equation}\label{eq:Msign}
\begin{aligned}
r_t \ge \binv \quad & \Longrightarrow \quad M \text{ has sign sequence } (-) \\
r_t \le \bfwnorm \quad & \Longrightarrow \quad M \text{ has sign sequence } (+) \\
r_t \in (\bfwnorm, \binv) \quad & \Longrightarrow \quad M \text{ has sign sequence } (+-) \text{ or } (+).
\end{aligned}
\end{equation}
\textbf{Essentially, the mistake in \cite{KR08} was to ignore the possible sign sequence $(+)$ in the third case.} Not repeating the same mistake, we take a closer look at the third case and note that the sign sequence of $M$ is $(+-)$ if and only if
\begin{equation}\label{eq:Mneg}
\lim_{x \to \infty} M(x)  < 0.
\end{equation}
Decomposing $M(x) = L_1(x) + r_t L_2(x)$ it remains to study the asymptotic properties of $L_1$ and $L_2$. We have
\begin{align*}
L_1(x) &= A(x) - x F(B(x)) = \int_0^x \left(F(B(s)) - F(B(x))\right)ds = \\
&= \int_0^{B(x)} \frac{F(u) - F(B(x))}{R(u) - 1}du \quad \xrightarrow{x \to \infty} \quad \int_0^c \frac{F(u) - F(c)}{R(u) - 1} du.
\end{align*}
In addition
\begin{align*}
L_2(x) &= B(x) - x \left[R(B(x)) - 1\right] = \int_0^x \left(R(B(s)) - R(B(x))\right)ds = \\
&= \int_0^{B(x)} \frac{R(u) - R(B(x))}{R(u) - 1}du \quad \xrightarrow{x \to \infty} \quad \int_0^c \frac{R(u) - 1}{R(u) - 1} du = c.
\end{align*}
Since $c < 0$, we conclude that 
\begin{equation}\label{eq:Minfty}
 \lim_{x \to \infty }M(x) < 0  \quad \Longleftrightarrow \quad  r_t > \bynorm =  \frac{1}{c}\int_c^{0} \frac{F(u) - F(c)}{R(u) - 1} du.
\end{equation}
By convexity of $F$ and $R$ and using that $c < 0$ we observe that
\begin{equation*}
\bynorm = \frac{1}{c}\int_c^{0} \frac{F(u) - F(c)}{R(u) - 1} du  \ge  \frac{1}{c}\int_c^{0} \frac{F'(c)}{R'(c)} du = - \frac{F'(c)}{R'(c)} = \bfwnorm.
\end{equation*}
Together with \eqref{eq:Msign} this completes the proof.
\end{proof}

\begin{proof}[Proof of Corollary~\ref{cor:order}] Recall that $R(c) = 1$ and $c < 0$. By convexity of $F$ and $R$ we have
\begin{equation}\label{eq:ineq}
\begin{aligned}
F'(c) &\le \frac{F(u) - F(c)}{u-c} \le \frac{F(c)}{c} \le F'(0)\\
R'(c) &\le \frac{R(u) - 1}{u-c} \le \frac{1}{c} \le R'(0).
\end{aligned}
\end{equation}
for all $u \in (c,0)$. Note that by (P1') either $F$ or $R$ is strictly convex, such that strict inequalities must hold in either the first or the second line. If $R'(0) < 0$, then applying the strictly increasing transformation $x \mapsto -\frac{1}{x}$ to the second line in \eqref{eq:ineq} and multiplying term-by-term with the first, we obtain
\begin{equation}
-\frac{F'(c)}{R'(c)} < -\frac{F(u) - F(c)}{R(u)-1} < -\frac{F(c)}{c} < -\frac{F'(0)}{R'(0)}.
\end{equation}
Applying the integral $\frac{1}{c}\int_0^c \,du$ to all terms, \eqref{eq:serial_ineq} follows. If $R'(0) \ge 0$, this approach is still valid for the first two inequalities in each line of \eqref{eq:ineq}, but not for the last one. However, in the case $R'(0) \ge 0$ we have set $\binv = +\infty$ in \eqref{eq:k0}, and the last inequality in \eqref{eq:serial_ineq} holds trivially. 
It remains to show that $D \cap (\bynorm, \binv)$ is non-empty. $F$ is a convex function and by Condition~\ref{cond1} finite at least on the interval $(c,0)$. It follows that $F'(0) > -\infty$ and thus that $\binv > -\infty$ in general. If $D = \Rplus$ then $F'(0) > 0$ by (P3) and hence $\binv > 0$. Moreover $\bynorm < \basymp = -F(c) < \infty$, completing the proof.
 \end{proof}

\bibliographystyle{alpha}
\bibliography{references}
\end{document}